\documentclass[11pt]{article}

\usepackage{parskip}
\usepackage[T1]{fontenc}
\newcommand{\blind}{0}

\usepackage{caption} 
\usepackage[margin=1in]{geometry}
\usepackage{setspace}

\usepackage{amsmath,amssymb,amsthm,bm}
\usepackage{graphicx,psfrag,epsf}
\usepackage{epigraph,xcolor}
\usepackage{enumerate}
\usepackage{xurl}
\usepackage{amsmath,latexsym,amssymb,bm}
\usepackage{hyperref}
\usepackage{bbm}
\usepackage{mathtools}
\usepackage{comment}
\usepackage{mathrsfs}
\hypersetup{colorlinks=true,citecolor=blue}

\usepackage{cleveref}

\usepackage{natbib}
\def\E{\mathbb{E}}
\newcommand{\convD}{\xrightarrow[]{\mathcal{L}}}
\newcommand{\convP}{\xrightarrow[]{\mathcal{P}}}
\newcommand{\sgn}{\mathrm{sgn}}
\newcommand{\var}{\mathrm{Var}}

\newcommand{\cov}{\mathrm{Cov}}
\newcommand{\ber}{\mathrm{Bernoulli}}
\newcommand{\bin}{\mathrm{Binomial}}
\newcommand{\expdist}{\mathrm{Exponential}}
\newcommand{\betadist}{\mathrm{Beta}}
\newcommand{\ind}{\mathbb{I}}

\newcommand{\gauss}{\mathcal{N}}
\renewcommand{\P}{\mathbb{P}}

\renewcommand{\vec}[1]{\mathbf{#1}}

\newcommand{\RNum}[1]{\uppercase\expandafter{\romannumeral #1\relax}}

\newcommand{\abs}[1]{\left\lvert #1 \right\rvert}

\newcommand{\paren}[1]{\mathopen{}\left( {#1}_{{}_{}}\,\negthickspace\right)\mathclose{}}
\newcommand{\bracket}[1]{\mathopen{}\left[ {#1}_{{}_{}}\,\negthickspace\right]\mathclose{}}

\numberwithin{equation}{section}  

\newtheoremstyle{general}
{3mm} 
{3mm} 
{\it} 
{} 
{\bfseries} 
{.} 
{.5em} 
{} 

\theoremstyle{general}

\newtheorem{lemma}{Lemma}
\newtheorem{theorem}{Theorem}
\newtheorem{corollary}{Corollary}
\newtheorem{assumption}{Assumption}
\newtheorem{remark}{Remark}
\newtheorem{prop}{Proposition}

\makeatletter
\renewenvironment{proof}[1][\proofname]{\par
    \pushQED{\qed}%
    \normalfont \topsep6\p@\@plus6\p@\relax
    \trivlist
    \item\relax{
        \bfseries
        #1\@addpunct{.}}\hspace\labelsep\ignorespaces
    }{%
     \popQED\endtrivlist\@endpefalse
     }
\makeatother

\begin{document}


\if0\blind
{
  \title{\bf Effect of influence in voter models and its  application in detecting significant interference in political elections}
  \author{Manit Paul\thanks{Email: manit282000@gmail.com}\hspace{.2cm} \\ 
  Indian Statistical Institute Kolkata \\ 203 Barrackpore Trunk Road, Kolkata, WB 700108, India. \\
  and \\
  Rishideep Roy\thanks{Email: rishideep.roy@iimb.ac.in}\hspace{.2cm} \\
  Indian Institute of Management Bangalore \\ Bannerghatta Main Road, Billekahalli, Bangalore, KA 560076, India. \\
  and \\
  Soudeep Deb\thanks{Email: soudeep@iimb.ac.in}\hspace{.2cm} \\
  Indian Institute of Management Bangalore \\ Bannerghatta Main Road, Billekahalli, Bangalore, KA 560076, India.}
  \maketitle
} \fi

\if1\blind
{
  \bigskip
  \bigskip
  \bigskip
  \begin{center}
    {\LARGE\bf Effect of influence in voter models and its application in detecting significant interference in political elections}
\end{center}
  \medskip
} \fi

\bigskip
\begin{abstract}
In this article, we study the effect of vector valued interventions in votes under a binary voter model, where each voter expresses their vote as a $0-1$ valued random variable to choose between two candidates. We assume that the outcome is determined by the majority function, which is true for a democratic system. The term intervention includes the cases of counting errors, reporting irregularities, electoral malpractice etc. Our focus is to analyze the effect of the intervention on the final outcome. We construct statistical tests to detect significant irregularities in elections under two scenarios, one where an exit poll data is available, and more broadly under the assumption of a cost function associated with causing the interventions. Relevant theoretical results on the consistency of the test procedures are also derived. Through detailed simulation study we show that the test procedure has good power and is robust across various settings. We also implement our method on three real life data sets. The applications provide results consistent with existing knowledge and establish that the method can be adopted for crucial problems related to political elections.


\end{abstract}

\noindent%
\vfill





\newpage

\section{Introduction} 
\label{sec:introduction}

In the modern era, opinions of individuals hold great power over decision-making at multiple levels, going up to the functioning of the government and the society. As such, the cumulative or interactive behavior within those opinions hold sway over these power structures. 
Opinion dynamics focuses on the way different options compete in a population, giving rise to either consensus (every individual holding the same opinion or option) or coexistence of several opinions. 
Our focus is particularly on the electoral system, and the effects of influence { (we shall use the term intervention interchangeably)} in them. The voter model has been studied extensively as an opinion dynamics model in this regard. We shall use the voter model in this paper for a two-party democratic system, which is prevalent in many countries, e.g.\ United States of America (USA). Our objective is to understand the change in the outcome of the vote when there is  intentional or unintentional external influence, such as counting errors, vote rigging, reporting issues etc. We shall primarily construct statistical hypothesis tests, based on exit poll data, to identify whether significant such interventions have happened. We also provide a cost function based approach to detect the same, under the situation when exit poll or opinion poll data is not available.

\subsection{Background and relevant literature}

Government forms a very important part of our society. There have been various forms of government at different times in our history. In ancient times, monarchy prevailed in our society. After that, came the era of autocracy, in which supreme power over a state is concentrated in the hands of one person, whose decisions are subject to neither external legal restraints nor regularized mechanisms of popular control. Thereafter, dictatorship showed its ugly presence in the society where the government was characterized by a single leader or group of leaders and little or no toleration for political pluralism or independent media. The evils of dictatorship was badly experienced by the society then, which finally paved the way for democracy. According to the Oxford dictionary, democracy is ``government by the people in which the supreme power is vested in the people and exercised directly by them or by their elected agents under a free electoral system''. In the phrase of Abraham Lincoln as mentioned in his biography, democracy is a government ``of the people, by the people, and for the people''. Our focus on this article is on an immensely crucial aspect of the democratic system. 

In most of the democratic countries across the globe today, the government is elected based on which candidate has received the maximum number or the \textit{majority} of votes. Majority rule is the binary decision rule used most often in influential decision-making bodies, including all the legislatures of democratic nations. Several works have been done on majority rule, see for example, \cite{messner2004voting}, \cite{hastie2005robust} and \cite{dasgupta2008robustness}. According to \cite{may1952set}, majority rule is the only binary decision rule that has the following properties: fairness in terms of anonymity and neutrality, decisiveness and monotonicity. Other forms of binary decision rules do not satisfy the above. For example, another commonly discussed decision function is the dictator function, which works under the assumption that the election result is completely determined by the choice of one person, known as the `dictator'. This rule does not satisfy the aforementioned properties. Refer to \cite{aswal2003dictatorial}, \cite{chatterji2014random} and \cite{pichler2018dictator} for relevant reading on dictator functions. 

An important issue that comes with the democratic election system and the majority rule is the problem of electoral irregularities, i.e.\ intentional or unintentional interference with the process of an election, subsequently increasing the vote share of a particular candidate, depressing the vote share of rival candidates, or both. In \cite{gibbard1973manipulation} and \cite{satterthwaite1975strategy}, it has been discussed how it is impossible to ensure that electoral systems are completely devoid of manipulations. In fact, the occurrences of electoral frauds is not at all uncommon in our society. For instance, in Georgia, an election server central to a legal battle over the integrity of Georgia elections was left exposed to the open internet for at least six months (\cite{Bajak20}).
There are several other instances of electoral fraud in other countries in modern times, see for example, \cite{mccann1998mexicans}, \cite{lehoucq2003electoral}, \cite{paniotto2004ukraine}, \cite{prado20112004}, \cite{casimir2013electoral}, and \cite{onapajo2014rigging}.

In light of the above, it becomes extremely important to come up with methods of identifying the occurrence of electoral anomalies. One of the most popular approaches in this regard is the Benford law. The works by \cite{deckert2011benford}, \cite{pericchi2011quick} and \cite{beber2012numbers} are some relevant studies. \cite{kobak2016integer} also hypothesized that if election results are manipulated or forged, then, due to the well-known human attraction to round numbers, the frequency of reported round percentages can be increased. This hypothesis was tested by analyzing raw data from seven federal elections held in the Russian Federation during the period from 2000 to 2012. \cite{rozenas2017detecting}, on the other hand, used a technique based on resampled kernel density methods to detect irregularities in the pattern of vote-shares. During the last decade, machine learning techniques have also been used for detecting election anomalies. \cite{cantu2010supervised}, \cite{medzihorsky2015election}, \cite{zhang2019election} and the references therein are some recommended readings in this context.  

\subsection{Our contribution}

All of the above-cited studies make various attempts to detect the presence of electoral irregularities. However, not much concentration has been given to confirm whether the interference actually causes a change in the true majority. To explain this further, consider a two-candidate election where the first candidate is likely to get more than 70\% votes. In this scenario, even if the second candidate intentionally intervenes up to 20\% votes, the majority function does not change. We shall call this type of influences `insignificant'. Along the same line, an influence is termed `significant' if there is a high probability of a shift in majority on the commitment of the intervention. Our focus in this paper is to develop statistical tests which can detect the presence of significant electoral intervention under various real-life scenarios. 

This new method of testing whether any malpractice has occurred in the election is developed under two scenarios. More broadly, we consider the scenario where we are given only the final election result. Here, we make appropriate assumptions about the cost associated with { intervention} to develop the test. In a more specific case, in addition to the final election result, we also have an exit poll data that would help us in drawing relevant inference. The main focus of our paper is this scenario with the added information of prior data. It is well-established that an exit poll can give an early indication as to how an election has turned out, especially because the counting process can be very time consuming in many elections. Polling is the only way of collecting pertinent information as the voters are anonymous. Exit polls have been historically used throughout the world to identify the degree of potential election fraud. Some examples of this are the 2004 Venezuelan recall referendum, and the 2004 Ukrainian presidential election, both of which will be discussed in greater detail later in this paper. We also point out that opinion polls can provide prior information as well, albeit they serve as a much weaker predictor of the election result, primarily because they are carried out before the election takes place. 



As an application of the proposed approach, we first perform an in-depth analysis for the $2016$ USA Presidential Elections. We check for the presence of statistically significant intervention in each of the states of USA. The exit poll data is obtained from \cite{ortiz2016exitpoll} and the final-election result is obtained from \cite{data2018}. \Cref{fig1} shows the exit poll results and the final election results side-by-side. It is interesting to note that Michigan, Nevada, North Carolina, Pennsylvania, Wisconsin are the only five states where the two results do not match. We  shall use our proposed methodology to investigate these five states in more detail. 

\begin{figure}[!htb]
    \centering
    \includegraphics[width=\textwidth,keepaspectratio]{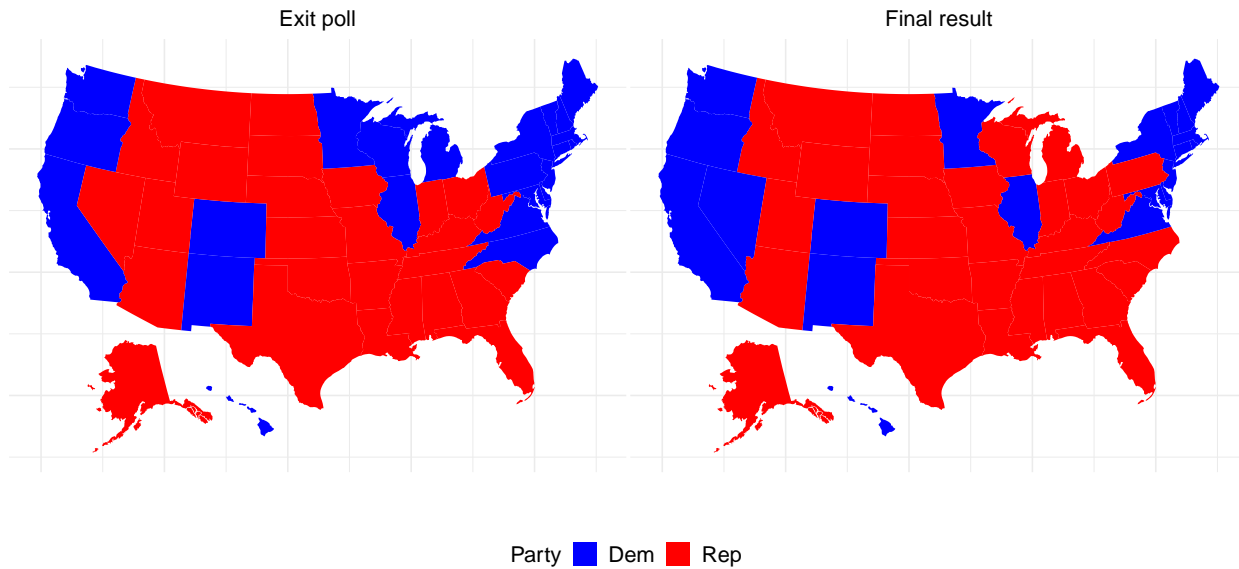}
    \caption{Comparison of Final election result and exit poll predictions for $2016$ USA Presidential election. Dem and Rep refer to the Democratic Party and Republican Party respectively.}
    \label{fig1}
\end{figure}

To further demonstrate the usefulness of the proposed approach, we perform statistical analysis for testing the presence of significant intervention in the $2004$ Ukrainian Presidential election and $2004$ Venezuelan recall referendum, both being known to have experienced electoral frauds. The official final election results and the exit poll data for these two studies are obtained from \cite{paniotto2004ukraine} and \cite{prado20112004}, respectively.

\subsection{Organization}

The paper is divided into the following sections. \Cref{sec:model} discusses the voter model in depth, as we introduce the intervention in a formal way. The main results in this section are the distribution of votes after intervention. Later in the section we compute the distribution of the intervened maximum. Next, \Cref{rigging} provides the method for testing for anomalies in elections under the two scenarios mentioned in the previous section.
In \Cref{sec:simulation}, we perform a detailed simulation study under various scenarios and discuss the findings. Real data applications are presented in \cref{sec:application}, while some concluding remarks and scopes of future research are summarized in \Cref{sec:conclusions}. 


\section{Model}
\label{sec:model}

Throughout this article, $\ind\{\cdot\}$ denotes the indicator function, i.e.\ $\ind\{A\}=1$ if $A$ is true and is 0 otherwise. For a real number $a$, $\sgn(a)$ is the signum function which takes the value $1, 0, -1$ depending on if $a$ is positive, zero or negative. The notations $\convP$ and $\convD$ indicate convergence in probability and convergence in law (distribution), respectively. The inner product of two vectors is denoted by $\langle \cdot,\cdot \rangle$. A $k$-variate normal distribution with mean $\theta$ and dispersion matrix $\Psi$ is denoted by $\gauss_k(\theta,\Psi)$. For a univariate normal distribution, we drop the subscript $k$ for convenience.
 
Our work is motivated primarily by electoral systems involving two candidates. Along the lines of \cite{condorcet1785essay}, we assume multiple voters who cast their preferences independently. We also assume that there is an overall popularity of each of the candidates, given by a proportion. Since there are two candidates, the sum of the two proportions would add up to one. Let us use $p$ to denote the overall popularity of the first candidate before the votes are cast, and we treat it as a random variable. Clearly, if $p_0$ is a realization of $p$, the proportion of votes in favour of the second candidate is $(1-p_0)$. These values are true parameters by which the voters independently exercise their choices. Every voter can choose one, and only one candidate. Henceforth, the opinion of each voter is given by a two-dimensional vector, with one entry as $1$ and the other as $0$. If a voter chooses the first candidate then it takes the value $(1,0)$, and otherwise it would be $(0,1)$. We assume in total there are $n$ voters. The mathematical model depicting this scenario is given in the next paragraph. 

Let $X_i$ be a two-dimensional vector which denotes the initial opinion of the $i^{th}$ voter, for $i=1,\hdots,n$. Throughout, we assume that $X_i$'s are independent and identically distributed (iid). We use $X_i=(1,0)$ (respectively $(0,1)$) if the $i^{th}$ voter originally supports the first candidate (respectively the second candidate). Thus, initially the distribution of $X_i$ is given as follows:
\begin{equation}\label{eq:distribution-xi}
    \P(X_i=t)=\begin{cases} p_0 & \text{ when } t=(1,0), \\
    1-p_0 & \text{ when } t=(0,1).
    \end{cases} 
\end{equation}


Let us now introduce the notion of intervention on votes. This is a form of outside influence. The primary motivation behind the idea of intervention is to study whether any electoral malpractice has happened during or after the casting of votes and before the vote-counting. Our objective is to statistically test whether the observed results are significantly different from what would have happened in the absence of this external influence. We assume a fixed form of this influence. In modeling this, we follow the notion of intervention introduced in \cite{hazla2019geometric}. Mathematically, it is denoted by a vector with all entries positive and is a transformation (calculated by the inner product) applied to all the voters. Depending on the relative magnitude of the entries which quantify the strength of the intervention, it may or may not alter the votes. In what follows, the probability of any voter supporting the first candidate after the intervention is assumed to be $p'$, as opposed to its original value of $p_0$ before intervention. The main results of this section are on the distribution of individual votes after the application of an intervention. 


Let the intervention vector $v=(\alpha,\beta)$, with $\alpha,\beta >0,$ be applied to each voter with probability $\pi_0$. Then, there are two cases:
\begin{itemize}
    \item If $v$ acts on $(1,0)$, the intervened vector is given by, $(1,0)+\langle\paren{1,0},\paren{\alpha,\beta}\rangle(\alpha,\beta)$, which is same as $(\alpha^2+1,\alpha\beta)$.
    \item If $v$ acts on $(0,1)$, the intervened vector is given by, $(0,1)+\langle\paren{0,1},\paren{\alpha,\beta}\rangle(\alpha,\beta)$, which is same as $(\alpha\beta,\beta^2+1)$.
\end{itemize}

Once we have the intervened vector, it is transformed to an opinion vector by looking at the maximum value between the two coordinates. To elaborate, we shall say that the updated opinion vector is $(1,0)$ if the first coordinate of the intervened vector is greater than the second one and is $(0,1)$ otherwise.

Note that the effect of the intervention in switching a vote depends on the magnitude of $\alpha,\beta$. For instance, if $v$ acts on $\paren{1,0}$ and if $\alpha^2+1 \geqslant \alpha\beta$, then the opinion vector remains the same i.e.\ $\paren{1,0}$. However, if $\alpha^2+1 < \alpha\beta$, the opinion vector switches to $\paren{0,1}$. Similarly, if the intervention acts on $\paren{0,1}$ and $\beta^2+1 \geqslant \alpha\beta$, then the opinion stays the same and otherwise, it switches to $\paren{1,0}$. 


\begin{lemma}
It is impossible to have an intervention that changes the opinion vector for all voters. 
\end{lemma}

\begin{proof}
Suppose that there exists an intervention vector $v=\paren{\alpha,\beta}$, which switches both $\paren{1,0}$ to $\paren{0,1}$ and $\paren{0,1}$ to $\paren{1,0}$. Then, we must have: 
\begin{equation}\label{eq:intervention-cond}
\alpha^2+1<\alpha\beta \text{ and }\beta^2+1<\alpha\beta,
\end{equation}
subsequently implying $(\alpha-\beta)^2+2<0$, which is impossible for real $\alpha,\beta$. Hence all opinions upon which the intervention acts cannot get switched. 
\end{proof}

The implication of the above lemma is pivotal for the following discussion, as we aim to develop a test for detecting irregularities in an election. It establishes that there cannot be an intervention vector which would corrupt the true opinion of all the voters. We discuss next to what extent an intervention can influence the choice of a particular voter. 

Consider the following notations. As mentioned earlier, there are $n$ voters under our consideration, whose initial opinion vectors are given by $X_i$, for $i=1,\hdots,n$. The probability that a particular voter supports the first candidate is denoted as $p_0$. In other words, all $X_i$'s are assumed to be iid $\ber(p_0)$ type random variables, where $p_0$ denotes the probability of the event $(1,0)$. Now, suppose that $\pi_0$ proportion of voters have been acted upon by an intervention $v=\paren{\alpha,\beta}$, where $\alpha,\beta \geqslant 0$, and let the updated opinion vectors be given by $X_i'$, for $i=1,\hdots,n$.  We should note that the proportion of voters who have been acted upon by the intervention $v$ is a random variable $\pi$ whose support is on $[0,1]$. Here, $\pi_0$ is a realization of $\pi$. Also, let $p'$ be the true post-intervention probability of supporting the first candidate for a randomly selected voter. In line with the above, it is easy to argue that all $X_i'$ are iid $\ber$ type random variables with parameter $p'$. The relationship of $p'$ with $p_0$ and $\pi_0$ is discussed in the following result.

\begin{lemma}\label{lem:intervention-vote}
Depending on the values of $\alpha,\beta$, there are three possible cases: (a) if $\alpha^2+1 \geqslant \alpha\beta$, $\beta^2+1 \geqslant \alpha\beta$, $p'=p_0$; (b) if $\alpha^2+1 < \alpha\beta$, $\beta^2+1 \geqslant \alpha\beta$, $p'=p_0-p_0\pi_0$; and (c) if $\alpha^2+1 \geqslant \alpha\beta$, $\beta^2+1 < \alpha\beta$, $p'=p_0+\pi_0-p_0\pi_0$.
\end{lemma}

\begin{proof}
For $\alpha^2+1 \geqslant \alpha\beta$, $\beta^2+1 \geqslant \alpha\beta$, we observed above that none of the opinion vectors changes, which proves part (a). 

Under the conditions of part (b), $v$ switches $\paren{1,0}$ to $\paren{0,1}$, but there is no change if $v$ acts on $\paren{0,1}$. Thus, $X_i'=X_i$ with probability $(1-\pi_0)$ and is equal to $(0,1)$ with probability $\pi_0$. It subsequently implies that
\begin{equation} 
\label{eq:case2}
\P\left[X_1'= (1,0)\right] = \P\left[X_1'=X_1\right]\P\left[X_1=(1,0)\right] = \paren{1-\pi_0}p_0.
\end{equation}

Hence, we see that in this case, $p'=p_0-p_0\pi_0$. Part (c) follows in a similar fashion by noting that $v$ can change only the opinion vector $(0,1)$.
\end{proof}

%



In the current work, our focus is on the democratic method of election, where the majority function determines the winning candidate. That is to say, the candidate who obtains more votes wins the election. Following the Condorcet Jury theorem in \cite{condorcet1785essay}, where voting is considered as an aggregation procedure and where the effectiveness of the majority opinion is shown asymptotically, we go with the same function here. 

For the following discussion, we make slight notational changes. Let us denote the opinion vector $\paren{1,0}$ by $1$ and the opinion vector $\paren{0,1}$ by $-1$, for convenience. Let $m\paren{\cdot}$ denote the majority function, i.e.\ $m\paren{\cdot}=1$ if the first candidate gets the majority and $-1$ otherwise. For $\vec{X}=\paren{X_1,X_2,\hdots,X_n}$ and $\vec{X'}=\paren{X_1',X_2',...,X_n'}$, we shall use $m(\vec{X})$ to denote the initial majority (before any kind of intervention) and $m\paren{\vec{X'}}$ to denote the post-intervention majority between the two candidates. 

The focus of this section is the computation of the distribution of the post-intervention majority, and how it is related to the original majority. Based on $p_0$ and $p'$ defined earlier, we want to calculate the probability of $m(\vec{X'})$ remaining same as $m(\vec{X})$. Naturally, the lower this probability is, the higher is the chance that the final outcome of the election changes because of the intervention. In \Cref{rigging}, we develop a statistical test based on this probability, the computation of which is given by the following proposition. The proof is deferred to \Cref{sec:proofs}.

\begin{prop}\label{thm:intervention}
Let $f\paren{\mu,\Sigma,n}$ be the probability that the two components of a $\gauss_2(\mu,\Sigma/n)$ distribution are of same sign. If the intervention vector is $v=\paren{\alpha,\beta}$ with $\alpha,\beta \geqslant 0$, then depending on the values of $\alpha,\beta$, the following results hold.
\begin{itemize}
    \item[(a)] If $\alpha^2+1 \geqslant \alpha\beta$, $\beta^2+1 \geqslant \alpha\beta$, $\P\paren{m\paren{\vec{X}}=m\paren{\vec{X'}}}=1$.
    \item[(b)] If $\alpha^2+1 < \alpha\beta$, $\beta^2+1 \geqslant \alpha\beta$, $\P\paren{m\paren{\vec{X}}=m\paren{\vec{X'}}}=f\paren{\mu_2,\Sigma_2,n} $.
    \item[(c)] If $\alpha^2+1 \geqslant \alpha\beta$, $\beta^2+1 < \alpha\beta$, $\P\paren{m\paren{\vec{X}}=m\paren{\vec{X'}}}=f\paren{\mu_3,\Sigma_3,n} $.
\end{itemize}
In the above, 
\begin{equation}
\label{eq:definition mu,sigma}
    \mu_2=\mu_3=\begin{pmatrix} 2p_0-1 \\ 2p'-1 \end{pmatrix}, \; \Sigma_2=\begin{bmatrix}4p_0-4p_0^2 & 4p'-4p_0p' \\
    4p'-4p_0p' & 4p'-4p'^2
    \end{bmatrix}, \;\Sigma_3=\begin{bmatrix}4p_0-4p_0^2 & 4p_0-4p_0p' \\
    4p_0-4p_0p' & 4p'-4p'^2
    \end{bmatrix}.
\end{equation}
\end{prop}



\section{Testing for irregularity in elections}
\label{rigging}

\Cref{thm:intervention} is key to develop a test for detecting significant interventions in elections. Recall the notations $p_0$ and $p'$ and suppose that $\hat{p}'$ is the observed proportion of votes in favour of the first candidate. Without loss of generality, we assume that the second candidate has won the election, i.e.\ $\hat{p}' < 0.5$. In other words, we develop the theory under the assumptions of part (b) in \Cref{thm:intervention}. The other cases would follow exactly similarly. 

Our objective is to check if there was any external influence involved in the victory of the second candidate. To put it in a more formal way, we wish to perform a test where the two hypotheses are framed as follows:
\begin{equation}
\label{eq:hypothesis}
    \begin{split}
        & H_0^n: \text{Significant electoral intervention has not occurred}, \\
        & H_1^n: \text{Significant electoral intervention has occurred}. 
    \end{split}
\end{equation}

We maintain the superscript $n$ to reflect the number of voters in the data. An intervention is termed as `significant' if there is a high probability of the majority being changed following the intervention i.e.\ if $f(\mu_2,\Sigma_2,n)<\tau_{c}$, for some pre-defined critical value $\tau_c \in (0,1)$. Generally, for simulation studies and for application to real data, we shall use $\tau_c=0.5$. As mentioned before, we work under two different frameworks and they are described in the following subsections. 

 

 
\subsection{Test procedure without prior data}
\label{sec:final data available}
 
Under the aforementioned assumptions, note that $\P(m(\vec{X})=m(\vec{X'}))=f(\mu_2,\Sigma_2,n)$, which is a function of $p_0,p'$ and $n$. Equivalently, we can also think of it as a function of $\pi_0,p'$ and $n$. Thus, we can say that, $f(\mu_2,\Sigma_2,n)=\eta(p_0,p',n)=\xi(\pi_0,p',n)$ for suitable functions $\eta(\cdot)$ and $\xi(\cdot)$. We wish to find a $100(1-\alpha) \%$ confidence interval for $f(\mu_2,\Sigma_2,n)$. This confidence interval can then be used to construct the test for detecting irregularity in the election.
 
Since no representative data on $p_0$ or $\pi_0$ is available, in order to develop the test procedure, we assume that the prior distribution of the random variable $\pi$ is known (note that as already mentioned in \Cref{sec:model}, $\pi_0$ is a realization of $\pi$). Let the probability density function (pdf) of $\pi$ be given by $h(\cdot)$, a continuous function supported on $\bracket{0,1}$, and let $H(\cdot)$ denote the corresponding distribution function (cdf). This density $h(\cdot)$ will be referred to as \textit{cost function} in our paper. One can interpret it as the cost associated with the intervention, which typically would be an increasing function of the proportion of votes influenced by the second candidate. In other words, the probability that $\pi$ is high should be quite low. From \Cref{lem:intervention-vote}, we know that, ideally, $p_0-(p'/(1-\pi_0))=0$. Since $\pi$ is distributed as $h(\cdot)$ in this case, we take into account the following assumption.

\begin{assumption}\label{asm-originalprop}
The cost function $h(\cdot)$ satisfies the following:
\begin{equation}
    \int_{0}^{1-p'}\paren{\frac{\pi_0-\pi}{1-\pi}}h\paren{\pi}d\pi=0.
\end{equation}
\end{assumption} 

It can be argued (shown in the proof of \Cref{thm:confidence-interval} in \Cref{sec:proofs}) that \Cref{asm-originalprop} implies that, for large population size, the original proportion of voters voting for the first candidate ($p_0$) matches with the expected value of $p$ when the final proportion of voters voting for the first candidate is $\hat{p}'$. Naturally, it is a sensible assumption which attempts to estimate $p_0$ by leveraging appropriate cost functions. The cost function $h(\cdot)$ can be a decreasing pdf like that of a truncated exponential distribution with parameter $\lambda$ or like a beta distribution with parameter $(1,\beta)$. 

Since the initial $p_0$ of the population before any intervention is unknown, we try to estimate $p_0$ from $\hat{p}'$ using the cost function. To that end, define
\begin{equation}
    \label{eq:definition_P}
    \phi\paren{{p}'}=\frac{{p}'}{H\paren{1-{p}'}}\int_{0}^{1-{p}'}\frac{1}{1-\pi}h\paren{\pi}d\pi.
\end{equation}

We note that $\phi({p}')$ is the expected value of $p$ when the final proportion of voters voting for the first candidate (final election result) is ${p}'$. Hence, $\phi(\hat{p}')$ is an estimate of $p_0$ that we get from $\hat{p}'$ using the cost function, $h(\cdot)$. Let $\gamma_n(\hat{p}')=\sqrt{\hat{p}'(1-\hat{p}')}z_{\beta/2}/\sqrt{n}$, where $(1-\beta)^2=1-\alpha$, and $z_{a}$ is the upper-$a$ quantile of a standard normal distribution. We consider the following two intervals,
\begin{align}
    S_1\left(\hat{p}'\right) &= \left(\hat{p}'-\gamma_n\left(\hat{p}'\right),\hat{p}'+\gamma_n\left(\hat{p}'\right)\right),  \\
    S_2\left(\hat{p}'\right) &= \left(\phi\left(\hat{p}'\right)-\abs{\phi'\left(\hat{p}'\right)}\gamma_n\left(\hat{p}'\right),\phi\left(\hat{p}'\right)+\abs{\phi'\left(\hat{p}'\right)}\gamma_n\left(\hat{p}'\right)\right),
\end{align}

We shall prove in \Cref{sec:proofs} that $S_1\paren{\hat{p}'}$ and $S_2\paren{\hat{p}'}$ are the $100(1-\beta) \%$ confidence intervals for $p'$ and $\phi(p')$ respectively. The following theorem illustrates how to compute the $100(1-\alpha) \%$ confidence interval for $f(\mu_2,\Sigma_2,n)$ when the distribution of $\pi$ is known to us. We know that if any value lies outside the confidence interval of confidence coefficient $100(1-\alpha)\%$ that value is rejected at the level of significance $\alpha$. Keeping this in mind, we define our test statistics for the problem given by \cref{eq:hypothesis}. 

\begin{theorem} \label{thm:confidence-interval}
Suppose the density of $\pi$ is given by $h(\cdot)$, where $h(\cdot)$ is a continuous density supported on $\bracket{0,1}$, and the cdf of $\pi$ is given by $H(\cdot)$. Then, under \Cref{asm-originalprop}, $(m,M)$ is a $100(1-\alpha)\%$ confidence interval for $f(\mu_2,\Sigma_2,n)$, where
    \begin{equation}
        \label{eq:defn_m1_M1}
        \begin{split}
            & m=\inf_{\left\{p' \in S_1\paren{\hat{p}'}, \; p_0 \in S_2\paren{\hat{p}'}\right\}}\eta\paren{p_0,p',n}, \\
            & M=\sup_{\left\{p' \in S_1\paren{\hat{p}'}, \; p_0 \in S_2\paren{\hat{p}'}\right\}}\eta\paren{p_0,p',n}.
        \end{split}
    \end{equation}

Further, if we use the decision rule that the hypothesis $H_0^n$ (see \cref{eq:hypothesis}) is rejected for $M<\tau_c$, then it is a consistent test at level of significance $\alpha$. 
\end{theorem}

Thus, if we know the prior density of $\pi$ or if we can estimate it from past elections, we can perform the test for the presence of significant electoral intervention using \Cref{thm:confidence-interval}. However, in practice, due to insufficient information, it might be quite difficult to obtain an estimate of the density of $\pi$ in most situations. In \Cref{sec:simulation}, we shall conduct a detailed simulation study assuming various distributions on $\pi$. Examining the performance of our test under various scenarios, we subsequently recommend some default distributions on $\pi$ that one can use to perform this test so that the type-$1$ error is controlled and the power is good. We next turn our attention to see how the test discussed in \Cref{thm:confidence-interval} can be improved in case we have an exit poll data for the concerned election.

\subsection{Test procedure using exit poll data}
\label{sec:final and exit available}

Consider the scenario where both the final election result and some prior information of $p_0$ are available. If the prior density of $p$ is known to us or can be estimated from past elections, then one can easily modify the steps from \Cref{sec:final data available} to test for significant electoral intervention. The following remark captures this discussion and the proof follows directly from \Cref{thm:confidence-interval}.

\begin{remark}
\label{rem:remark_to_thm1}
Assume that for large population size, the true intervention probability $\pi_0$ matches with the expected value of $\pi$ when the final proportion of voters voting for the first candidate is $\hat{p}'$. Then, assuming that the prior distribution of $p$ is known or can be estimated, the test for the occurrence of significant electoral intervention can be performed in the same spirit as \Cref{thm:confidence-interval}. One only needs to replace $\eta\paren{\cdot}$ with $\xi\paren{\cdot}$ and $\phi(\hat{p}')$ with the expected value of $\pi$ given the final proportion of voters voting for the first candidate is $\hat{p}'$ in the definitions of $S_1\paren{\hat{p}'}$ and $S_2\paren{\hat{p}'}$.
\end{remark}

From a pragmatic standpoint, obtaining an estimate of the prior density of $p$ might be tricky. However, we can leverage the exit poll data to conduct the test in a slightly different way which would obviate the need for estimating the density of $p$. It is assumed that the exit poll is drawn uniformly from the entire population. The total number of voters under consideration is $n$, and we consider that the exit poll data consists of a small portion of the total population. Let it be of size $k < n$. Clearly, for a fixed $n$, the accuracy of the testing method should improve with a larger value of $k$. 

Suppose, $\hat{p}_k$ is the observed proportion of voters voting for the first candidate in the exit poll data. We do not require any other assumption in this case. For the hypothesis testing problem in \cref{eq:hypothesis}, we can follow a similar procedure as in \Cref{thm:confidence-interval}. Define the test statistic $\mathscr{T}_1$ which is calculated identically to $M$ from \cref{eq:defn_m1_M1}, but with $S_2(\hat{p}')$ therein replaced by 
\begin{equation}
     C_2 \left({\hat{p}_k}\right) = \left(\hat{p}_k-\gamma_k \left(\hat{p}_k\right), \hat{p}_k+\gamma_k\left(\hat{p}_k\right)\right).
 \end{equation}



Note that $C_2(\hat{p}_k)$ provides a more accurate confidence interval for $p_0$ as it is based on a representative data of the entire population. It also explains why no other assumption is needed in this case. We can adopt the decision rule to reject $H_0^n$ if $\mathscr{T}_1<\tau_c$ i.e.\ if the probability of the majority remaining unchanged is small. By following identical steps as in \Cref{sec:final data available}, it can be shown that this results in a consistent level-$\alpha$ test too. The following corollary provides the necessary details. The proof is elaborated in \Cref{sec:proofs}.

\begin{corollary}\label{thm:exit poll theorem}
If we test the hypothesis given by \cref{eq:hypothesis} using the statistic $\mathscr{T}_1$ as outlined above, then it is a level-$\alpha$ test i.e.\ $\P_{H_0}(\mathscr{T}_1<\tau_c) \leqslant \alpha$. This is a consistent test as well.
\end{corollary}


\section{Simulation study}
\label{sec:simulation}

In this section,  we provide a detailed simulation study to illustrate the theoretical results discussed so far as well as to compute the type-$1$ error and the power of the test across various setups.  



\subsection{Performance of the test procedure without prior data}
\label{simulation_cost}

Consider the setup of \Cref{sec:final data available}. In practice, it is impossible to know the true distribution of the intervention probability $\pi$, and without prior knowledge or relevant data, one can only `guess' the nature of $h(\cdot)$. Let us denote this as $\hat h(\cdot)$. We have already shown that under \Cref{asm-originalprop}, if $\hat h(\cdot)$ is the same as true $h(\cdot)$, then the type-$1$ error remains under $\alpha$ and the power goes to $1$ as the number of voters goes to  infinity. Here we demonstrate the behavior of the type-$1$ error and the power when $\hat h(\cdot)$ is possibly different from $h(\cdot)$. 

Throughout, we use the population size $n=100,000$ and simulate data under \Cref{asm-originalprop}. Both $h(\cdot)$ and $\hat h(\cdot)$ are chosen from $\expdist(\lambda)$ and $\betadist(1,\beta)$, where $\lambda$ and $\beta$ are taken from the set $\{10,20,30,50,100\}$. For each such scenario, $200$ simulations are conducted and subsequently, empirical type-$1$ error and power are computed.  All tests are carried out at level $0.05$ using $\tau_c=0.5$ (in accordance with the notation used in \Cref{sec:final data available}). These results are shown in \Cref{fig2} and \Cref{fig4}. The columns in both figures represent the true distribution $h(\cdot)$ while the rows represent the assumed distribution $\hat h(\cdot)$. The parameter values used in the graphs refer to the rate parameter $\lambda$ (respectively the shape parameter $\beta$) if the distribution is exponential (respectively beta). Note that the parameter value of $h(\cdot)$ is given along the $x$-axis whereas the same for $\hat h(\cdot)$ is indicated by different colors. We also point out that for the aforementioned choices of the distributions, it is ensured that the probability of intervention on all or most voters is not substantially big. This is a realistic scenario as discussed earlier. 


\begin{figure}[!ht]
    \centering
    \includegraphics[width=\textwidth,keepaspectratio]{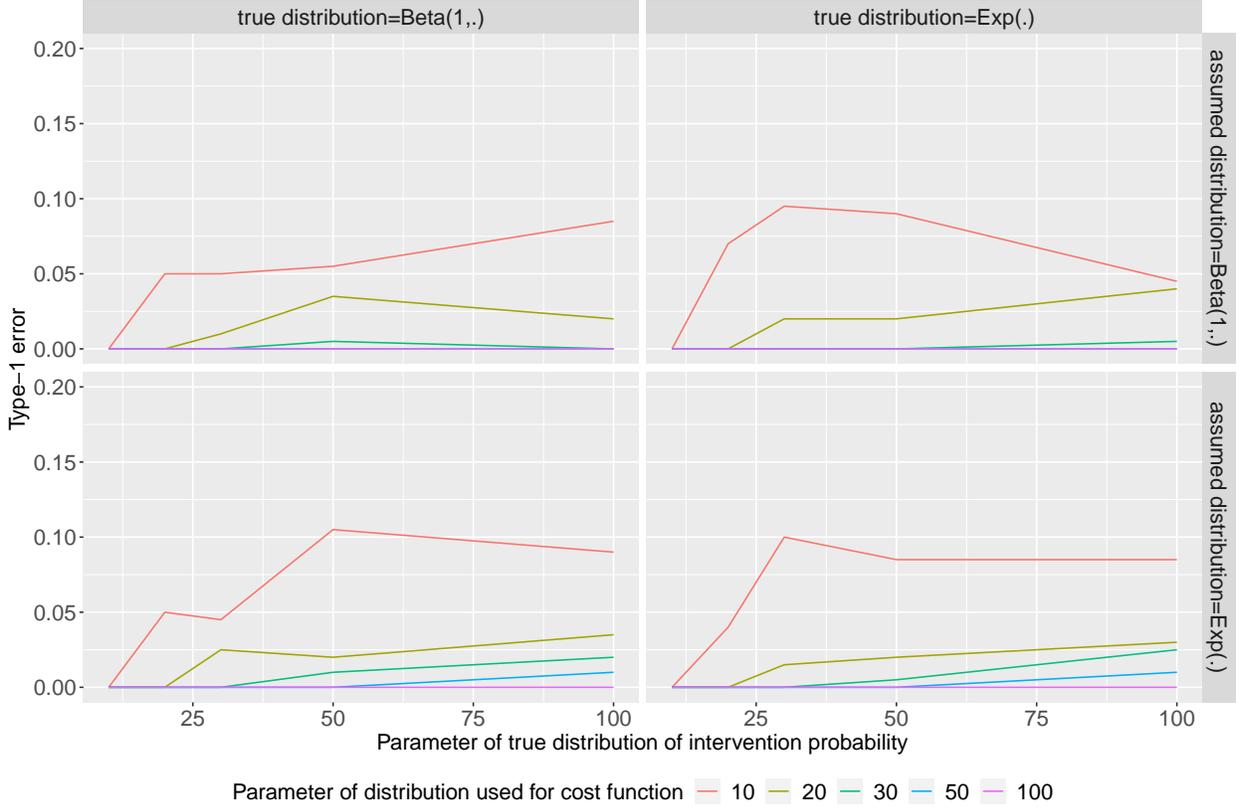}
    \caption{Empirical type-$1$ error of the test procedure with different cost functions.}
    \label{fig2}
\end{figure}

\begin{figure}[!ht]
    \centering
    \includegraphics[width=\textwidth,keepaspectratio]{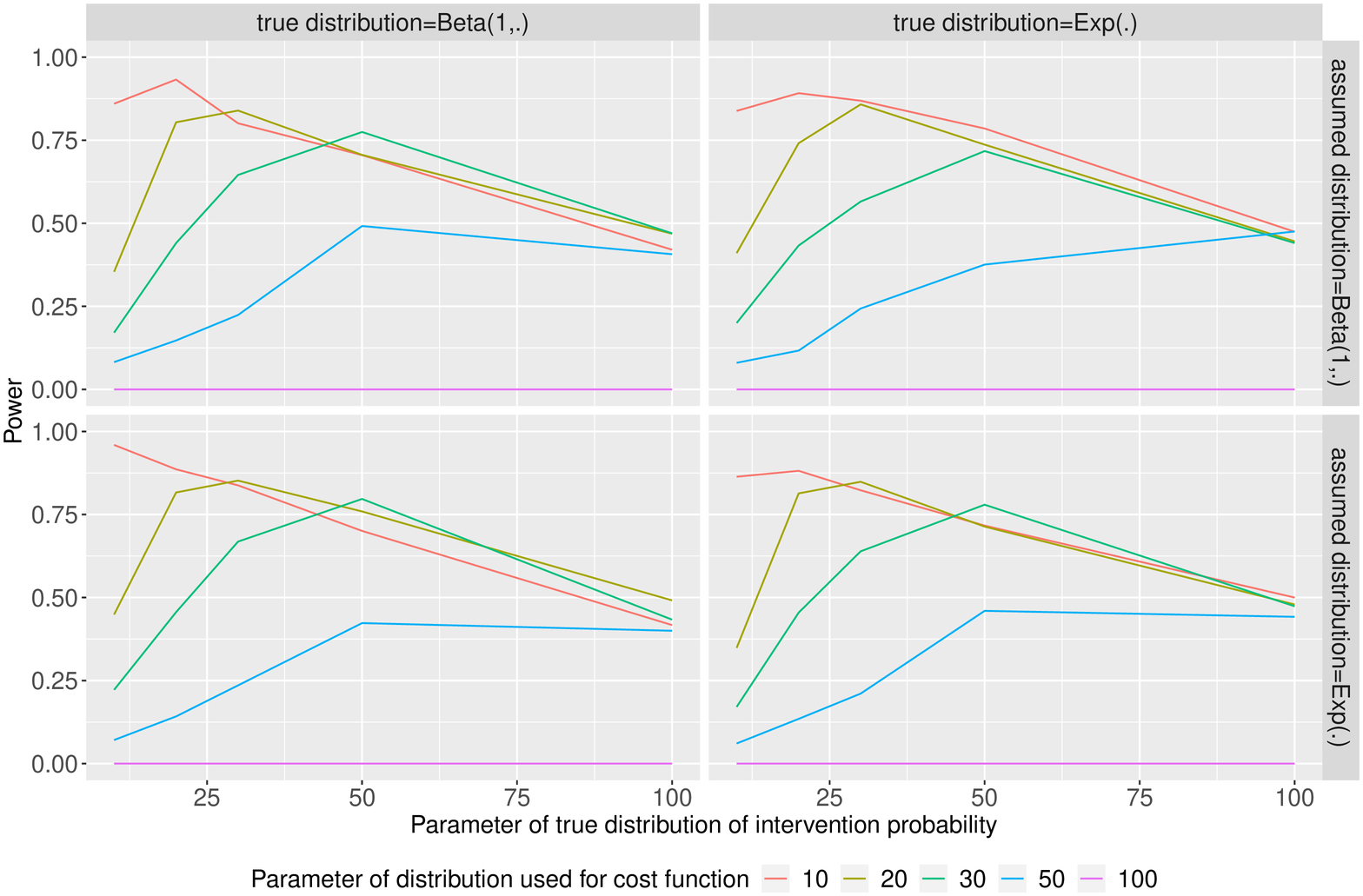}
    \caption{Empirical power of the test procedure with different cost functions.}
    \label{fig4}
\end{figure}

We note that both the type-$1$ error and power decrease with the increase in steepness (rate/shape parameter) of the distribution assumed for the cost function. Hence, higher the parameter of $\hat h(\cdot)$, lower is the chance of occurrence of any false positives. We observe from \Cref{fig2} that, when the parameter of $\hat h(\cdot)$ is above $20$, the type-$1$ error always stays below $0.05$. Thus, these can be possible choices for $\hat h(\cdot)$. We further observe in \Cref{fig4} that in all situations, the power is about 70\% or more if the parameter of the assumed distribution is around 20 to 30. It however becomes quite low when the parameter of $\hat h(\cdot)$ is $50$ or more. 

From the above, we can infer that if the original distribution of $\pi$ can be estimated through some prior knowledge, then one can use that in the cost function. This would keep the type-$1$ error below the chosen level $\alpha$ and give significantly high power. However, in case the prior density of $\pi$ is not available, and if there is no statistical way to estimate that density, then the parameter of $\hat h(\cdot)$ may be chosen to be a number between $20$ and $50$. In our real-life applications in \Cref{sec:application}, we shall assume that the distribution for the cost function is $\expdist(30)$ or $\betadist(1,30)$. Henceforth, the test statistic based on these two cost functions are referred to as $\mathscr{T}_2$ and $\mathscr{T}_3$, respectively.

\subsection{Performance of the test procedure using exit poll data}
\label{simulation_exit}

\Cref{sec:final and exit available} laid out the methodology of conducting the test for detection of electoral intervention when prior information about $p_0$, possibly obtained from an exit poll, is available. This subsection considers such cases and demonstrates how the power of the test varies with the size of the exit poll ($k$), the total population size ($n$), the true proportion of voters voting for the first candidate ($p_0$) and the value of the same after intervention ($p'$). 

The computation of the power is done under nine different combinations of $n$ and $k$. Three different sizes of total population are taken {\it viz.}\ $100,000$, $200,000$ and $1,000,000$. For each of these, $k$ is varied between $10,000$, $20,000$ and $50,000$. Now, for each of these nine combinations, 200 experiments are carried out and subsequently we calculate the empirical type-1 error and the empirical power of the test. The type-1 error, as expected, lies within the acceptable range and we omit those results for conciseness of the paper. The power values are plotted in \Cref{fig3} against the initial proportion of voters who voted for the first candidate. The color of the power plot represents the final proportion of voters who voted for the first candidate. We show the results for three different values of $p'$, {\it viz.}\ $0.45,0.49,0.495$, to understand the effect of the same on the power of the test. All of the experiments are conducted using the critical value $\tau_c=0.5$ and at level $0.05$. 

\begin{figure}[!ht]
    \centering
    \includegraphics[width=\textwidth,keepaspectratio]{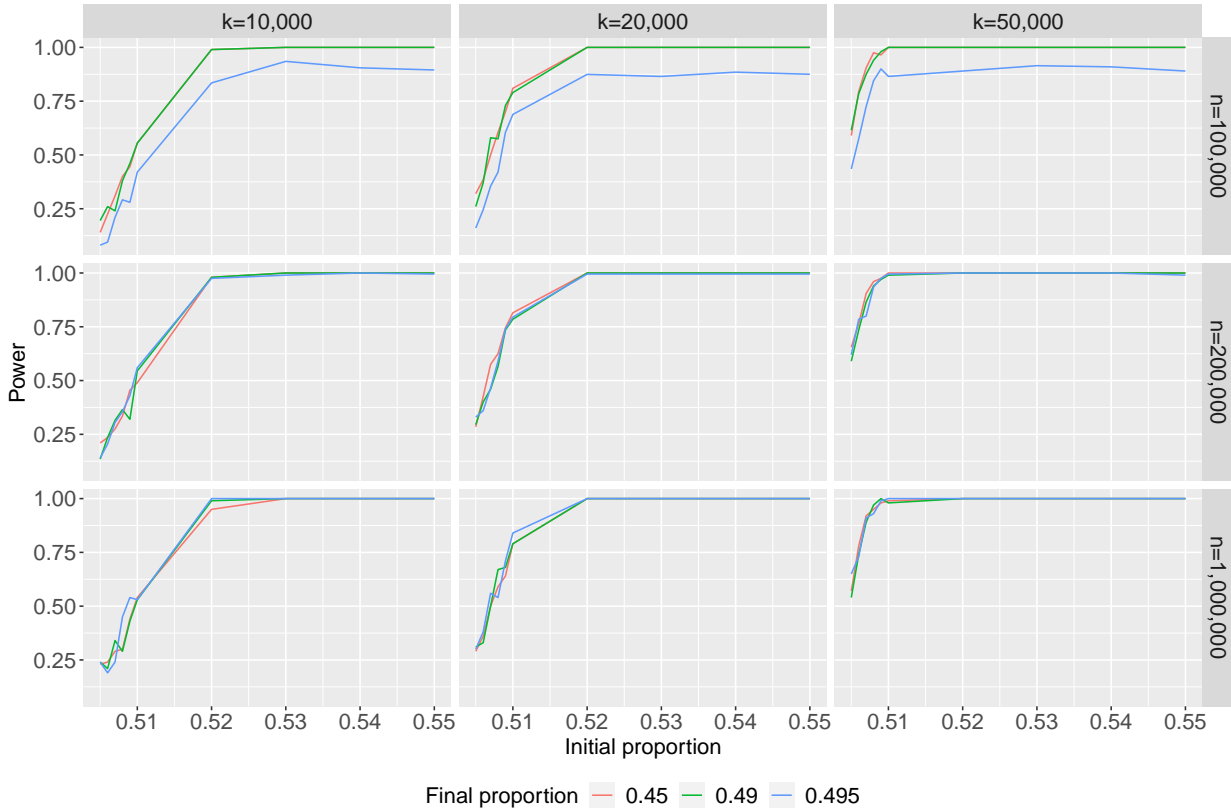}
    \caption{Power curves of the test for different combinations of initial and final (after intervention) proportion of voters voting for first candidate. $n$ represents the total population size and $k$ represents the size of the exit poll data.}
    \label{fig3}
\end{figure}

We observe from \Cref{fig3} that, if the size of the exit poll data is kept fixed, the power curves more or less stay the same. However, they change drastically when we change the exit poll size. They become steeper, and therefore better on more occasions, on increasing $k$. On the other hand, for $n=100,000$, the power curves are slightly different from the others. With total population size $200,000$ or above, we can see that the curves hardly vary. We can deduce some important conclusions from these observations. The power of the test primarily depends on the absolute value of the exit poll size and not so much on what proportion of the total population the exit poll data is. The test becomes more accurate with the increase in size of the exit poll data. Another thing to note is that when the total population size is large $(n \geqslant 200,000)$, the power curves of the three different colors almost coincide with each other. It suggests that the test does not depend significantly on the final proportion of the voters who voted for the first candidate, but mainly depends on the initial proportion of voters who would have voted for the first candidate. For an exit poll size of 20,000 or more, the test procedure has high power even when the true proportions of voters voting for the two candidates are very close (equivalent to saying $p_0$ is just above 0.5). In light of this, we can also argue that an exit poll data of size at-least $20,000$ can ensure accurate detection of significant influence with a very high probability. This can be explained theoretically by an application of usual central limit theorem which says that for sample size 20,000 or more, $\hat p_k$ can estimate $p_0$ within $\pm 1\%$ error with a very high probability.


\section{Real data application}
\label{sec:application}

\subsection{2016 USA Presidential Election}
\label{sec:us_election}

The $2016$ USA presidential election was the $58$th quadrennial presidential election, held on Tuesday, November $8$, $2016$ in which, the Republican candidate Donald Trump defeated the Democratic candidate Hillary Clinton. We select this data for a few reasons. First, it is the most recent presidential election in the USA in which the winning candidate lost the popular vote. Second, standard procedure of exit poll data collection took place in 2016, which is not the same due to the ongoing COVID-19 pandemic in 2020 USA Presidential Election. And most importantly, a popular theory suggested that illegal interference did take place in this election. On October 7, 2016, \cite{usintelligence} issued a joint statement that the intelligence community is confident that the Russian government directed the recent compromises of e-mails from a few USA persons and institutions, including some political organizations. A Special Counsel began in May $2017$ in order to investigate the alleged collusion between Russia and the republican party led by Trump. The counsel ended in March $2019$. According to \cite{mueller2019mueller}, the following conclusion was reached by the investigation: the Russians interfered ``in sweeping and systematic fashion'' to favor Trump's candidacy but it ``did not establish that members of the Trump campaign conspired or coordinated with the Russian government''.


Nonetheless, some discernible changes were observed in the outcomes of a few states. For the first time since $1984$, Wisconsin was won by the Republican party. Michigan and Pennsylvania were also won by them for the first time since $1988$. Jill Stein, the presidential candidate of the Green party, petitioned for a recount in these three states (\cite{steinrecount}). Around the same time, as reported by \cite{nevadareview}, Rocky De La Fuente, the presidential candidate of American Delta Party/Reform Party, was granted a partial recount in Nevada. Only minor changes to vote tallies were detected in the recounts in Wisconsin and Nevada, as reported in \cite{wisconsinnochange} and \cite{nevadanochange} respectively. A partial recount of Michigan ballot revealed some unbalanced precincts in Detroit and they were corrected as well. The state audit that followed came to the conclusion that the unbalanced precincts were a result of errors committed by the precinct workers and not a result of some major voter fraud (\cite{michigannochange}). Thus in spite of so many petitions, the recounts did not alter the outcome of the election. In the language of this article, we can say that these recounts did not identify significant electoral irregularities in any of the states. The above informations regarding the $2016$ USA Presidential election (in the above two paragraphs) were obtained from \cite{2016uselection} and then cross-verified from various news articles, reports that are cited above.

To statistically investigate the same problem through our proposed approach, we primarily use the exit poll data obtained from \cite{ortiz2016exitpoll} and the test statistic $\mathscr{T}_1$. Following the discussion in \Cref{sec:simulation}, we also check for robustness of our method by using the test statistics $\mathscr{T}_2$ and $\mathscr{T}_3$, which correspond to the cost functions $\expdist(30)$ and $\betadist(1,30)$ respectively. We point out that the exit poll data had multiple polls conducted for each state in various time points. Assuming that all the polls conducted for a particular state were disjoint, we combine them into a single exit poll. We note that the size of the exit poll is above $20,000$ for almost all the states, thereby ensuring high power of our test. Akin to the earlier sections, we use the critical value $\tau_c=0.5$. The individuals who neither voted for Trump nor for Clinton are removed from both the election result data and the exit poll data so that we are in the set-up of two-candidates voter model. We rescale the proportions accordingly. 

On making a naive comparison of the exit poll data with the final election results, we observe that there are five states where the results from the two sources do not match. This has been discussed briefly in \Cref{sec:introduction} (see \Cref{fig1}) as well. These five states are Michigan, Pennsylvania, Wisconsin, North Carolina and Nevada. In the first four states, Clinton was predicted to win and in the fifth state, Trump was predicted to win according to the exit poll data. Naturally, it makes sense to investigate these five states in more detail. We however start with a succinct account of the results for all of the other states. The value of the statistics $\mathscr{T}_1,\mathscr{T}_2$ and $\mathscr{T}_3$ are all found out to be 1 for all of these states, which provides absolutely no indication of any type of statistical evidence for electoral irregularity to make a case for recounting in these states. It is critical to observe that even without the information from exit polls, as we adopt the cost function approach, the inference remains exactly the same.


We now take a detailed look at the aforementioned five states, where the value of the test statistics are found to be less than 1. Relevant details for these states, along with the value of the test statistic and the conclusion, are presented in \Cref{tab:us-election}. For analysis purposes, we mainly pay attention to $\mathscr{T}_1$, as it has already been shown to have much higher power and less type-$1$ error than the other two methods. 

\begin{table}[!ht]
\caption{Results for the electoral irregularity detection method for five states in 2016 USA Presidential Elections.}
\label{tab:us-election}
\centering
 \begin{tabular}{l c c c c c} 
 \hline
   & Michigan & Nevada & North Carolina & Pennsylvania & Wisconsin \\[0.5ex] 
 \hline
Population Size & 4,548,382 & 1,051,318 & 4,551,947 & 5,897,174 & 2,786,263 \\ 
 Exit Poll Size & 214,280 & 156,628 & 305,032 & 342,667 & 192,636\\

 Exit Poll (Clinton) & 51.88\% & 49.60\% & 50.11\% & 51.60\%  & 52.77\% \\ 
 
 Exit Poll (Trump) & 48.12\% & 50.40\% & 49.89\% & 48.40\% & 47.23\% \\ 
 Final Result (Clinton) & 49.88\% & 51.29\% & 48.09\% & 49.62\% & 49.59\% \\
 
 Final Result (Trump) & 50.12\% & 48.71\% & 51.91\% & 50.38\% & 50.41\% \\
 
 $\mathscr{T}_1$ (test statistic) & 0.00109 & 0.00106 & 0.99711 & $\approx 0$ & $\approx 0$ \\ 
 
 Decision & Reject & Reject & Do not reject & Reject & Reject\\ 
 \hline
\end{tabular}
\end{table}

We observe that the values of the test statistic $\mathscr{T}_1$ for the states of Michigan, Nevada, Pennsylvania and Wisconsin are very close to zero indicating the presence of significant electoral anomalies in these four states which may have led to a change in the true majority. This gives sufficient statistical evidence to make a case for recounting in these four states. On the contrary, a very interesting situation occurs for the state of North Carolina. Here, although the exit poll prediction does not match with the final election result, the value of our test statistic comes out to be $0.99711$. Hence, $H_0^n$ is not rejected for this state. Thus we cannot say that significant electoral anomaly has been observed to make a case for recounting in this state. It is worth specifying that the proportions in the two samples (exit poll data and final results data) are indeed significantly different; but our test procedure does not detect significant illegal interference to change the majority function. It strengthens the importance of our test statistic over mere comparison of the exit poll data and the final election result. This also substantiates the fact that unlike the other four states no claim or petition of recount was made for North Carolina. 

In conclusion, significant evidence of influence is detected at 5\% level of significance in only four states, thereby providing statistical evidence for recounting. Of these four states, in Nevada, the final majority went in favor of the Democratic Party whereas for Michigan, Pennsylvania and Wisconsin, it was the other way round. No significant electoral anomaly is detected in any of the other states. It is imperative to declare that the the test statistics $\mathscr{T}_2$ and $\mathscr{T}_3$ render identical results in all states. Clearly, the decisions based on the cost function approach are the same for the entire election data. It establishes that even without sufficient prior information, the proposed approach works well and furnishes concrete support to the results obtained in this application.

\subsection{Fraudulent Presidential Elections from Ukraine and Venezuela}

Our second application revolves around the $2004$ Ukrainian Presidential election and the $2004$ Venezuelan recall referendum, both of which are known instances of electoral irregularities. Therefore, it would allow us to objectively evaluate whether our test procedure is able to detect irregularities in similar settings. 

According to the electoral law of Ukraine, the President is elected by a two-round system in which a candidate must win a majority (50\% or more) of all ballots cast. The first round of voting in 2004 election was held on $31$st October. Since no candidate had 50\% or more of the votes cast, a run-off ballot was held on 21st November between the two candidates with maximum number of votes, Viktor Yushchenko and Viktor Yanukovych. The run-off election was won by the latter according to the official results announced on 23rd November. The results of the second round were protested by the opposition with an allegation of illegal falsifications. Massive street protests in support of the opposition as well as the blockade and picketing of the government buildings (the so-called ``Orange Revolution'') with the demands to cancel the results of the elections went off in the country. The Supreme Court annulled the November runoff election and ordered the third round of election (a rerun of the second round) which took place on 26th December, 2004. This time, Yushchenko won the election with 52\% of the votes. The above information and other pertinent discussions are reported in a detailed manner in \cite{paniotto2004ukraine}, \cite{d2005last} and \cite{kuzio2005kuchma}.

Analogous to the previous application, here also we primarily use the exit poll data to reconfirm the presence of interference in the outcome of the second round of $2004$ Ukrainian Presidential Election. The exit poll was originally conducted by KIIS and the Razumkov Center with the organizational support of the Democratic Initiatives Fund, and was carried out nationwide by a secret ballot method. The details of the poll can be obtained from \cite{paniotto2004ukraine}. The sample consisted of 750 polling stations with about 28000 respondents and a response rate of 79\%. Hence, the size of the exit poll data is roughly 22120, which ensures high power for the test. 

On the other hand, the Venezuelan Recall Referendum (RR) of 15th August 2004 was a referendum to decide whether or not the then President Hugo Chávez should be removed from office. In order to activate the RR, on 28th November, 2003, signatures and fingerprints were collected in a four-day event organized by the Consejo Nacional Electoral (CNE), with witnesses from all political parties. The Organization of American States (OAS) sent a delegation chaired by its Secretary General to negotiate a solution. The Carter Center, led by President Jimmy Carter himself, played an important role in getting the government and the opposition to agree on a course of action. CNE was the official body in charge of the organization of the RR. The result of the referendum was not to recall Chávez (approximately 59\% voted against him). However, there have been allegations of fraud shortly after (\cite{prado20112004}, \cite{mccoy20062004}). Since the RR was seen by all parties as a pivotal event, several organizations set up schemes to collect exit poll data to address the allegations. The two exit polls considered here were conducted independently by S\'{u}mate, a nongovernmental organization, and Primero Justicia, a political party. Assuming that the two polls were disjoint, we combine them to get a sample of 36174 observations. These data are obtained from \cite{prado20112004}. 

For convenience, for both of the datasets, let us denote the candidate who has won the final election as $W$ and the corresponding losing candidate as $L$. As before, we consider a two candidate setup and rescale the data as needed. Critical value of the test is $\tau_c=0.5$, and the test is carried out at 5\% level of significance. The details of the exit poll data, final election results, and the outcome of the test procedure through $\mathscr{T}_1$ are displayed in \Cref{tab:ukrivene-election}.

\begin{table}[!ht]
\caption{Results for the electoral irregularity detection method for 2004 Ukrainian Presidential election and 2004 Venezuelan recall referendum.}
\label{tab:ukrivene-election}
\centering
 \begin{tabular}{l c c } 
 \hline
   & Ukraine & Venezuela  \\[0.5ex] 
 \hline
 Population Size & 29,315,980 & 9,789,637 \\ 
 Exit Poll Size & 22,120 & 36,174 \\
 Exit Poll (L) & 54.64\% & 60.63\%  \\ 
 Exit Poll (W) & 45.36\% & 39.37\%  \\ 
 Final Result (L) & 48.51\% & 40.64\% \\
 Final Result (W) & 51.49\% & 59.36\% \\
 $\mathscr{T}_1$ (test statistic) & $\approx 0$ & $\approx 0$ \\ 
 Decision & Reject & Reject \\ 
 \hline
\end{tabular}
\end{table}

We observe that the value of $\mathscr{T}_1$ is 0 for both Ukraine and Venezuela, indicating significant interference on both elections. This reaffirms the occurrence of major electoral irregularities in these two elections as has been discussed in the papers alluded to above. We also point out that the values of the other two statistics $\mathscr{T}_2$ and $\mathscr{T}_3$, which rely on the cost function based approach, match with that of $\mathscr{T}_1$ for $2004$ Ukrainian Presidential Election. In other words, it demonstrates the robustness of the proposed test procedure once again.

\section{Conclusion}
\label{sec:conclusions}

In summary, this paper provides a new method for leveraging the exit poll data to detect the occurrence of significant intervention in the outcome of a two-candidate democratic electoral system. We must note that not all intervention would lead to a change in the majority function of the voter model. Therefore, simple comparison of the final election result and the exit poll data is likely to render misleading conclusions. To that end, we work with a probabilistic voter model and develop a test procedure with a solid theoretical understanding. Through a detailed simulation study, we demonstrate the performance of the proposed test statistic under various settings. It establishes that when the size of the exit poll data is more than 20000, the test achieves high power. However, one should remember that  a key assumption of the test is that the exit poll is conducted in a scientific way where the sample is representative of the entire population. Thus, the test may not give good results if the exit poll data involves biased selection of people from the population. In such cases, or in the absence of any prior information, we also provide a cost function based approach that can lead to powerful tests as well. 

As real life applications, we consider three different examples. In the first one, our method detects the presence of significant electoral intervention in four states of USA in the $2016$ Presidential Election. The fact that recounting was done in all of these states substantiates our findings. It is also observed that for one of the states, the conclusion of the proposed test differs from that of a two sample proportion comparison. It clearly exhibits that our method is able to identify significant intervention that affects the overall outcome of a electoral process, thereby avoiding misleading inferences. Apart from that, we also detect evidence of significant electoral fraud in $2004$ Ukrainian Presidential election and $2004$ Venezuelan recall referendum, both of which are in line with existing knowledge. Overall, the results suggest that the procedure works well and thus can be used by news channels, political analysts and others to detect the presence of significant electoral intervention. 

We conclude this article with a couple of interesting future directions. Albeit one can detect the presence of irregularities on the overall outcome of the election (for example, on state-level), it does not provide additional idea of which sub-unit (for example, which county) is more likely to have been exposed to the irregularities. It would be interesting to work on a unified approach which leverages individual county or district level data to detect if the result of the state has been significantly altered (note that the majority is decided based on the data from the entire state) and if so, which counties or districts might have caused that. Another possible extension to the current work is to consider the case of multiple candidates instead of a binary voter model. It would allow us to study various elections in several other countries where there are more than two major political parties competing  against each other. 





\section*{Funding details}

There are no funding bodies to thank relating to the creation of this article.

\section*{Declaration of interest} 

The authors declare no conflict of interest.


\section*{Data availability statement} 

All of the data used in this article are publicly available. Data related to the 2016 USA Presidential Election are obtained from Harvard Dataverse (link: \url{https://dataverse.harvard.edu/file.xhtml?fileId=4788675}) and from Data world (link: \url{https://data.world/databeats/2016-us-presidential-election}). Data related to the other applications are obtained from \cite{paniotto2004ukraine} and \cite{prado20112004}.



\section{Proofs}
\label{sec:proofs}

\begin{proof}[Proof of \Cref{thm:intervention}] For the first case, when $\alpha^2+1 \geqslant \alpha \beta$ and $\beta^2+1 \geqslant \alpha \beta$, by \Cref{lem:intervention-vote}, since none of the opinion vectors change on being acted upon by the intervention vector $v$, it is easy to argue that $\P(m(\vec{X})=m(\vec{X'}))=1$. 

Now, let us focus on the second case. With $\ind\{\cdot\}$ denoting the indicator function, we have the relation $\ind\{m(\vec{X})=m(\vec{X'})\} = (1+m(\vec{X})m(\vec{X'}))/2$, which subsequently implies the following:
\begin{equation}
    \label{eq:m=m'}
    \P\paren{m\paren{\vec{X}}=m\paren{\vec{X'}}}=\frac{1}{2}\paren{1+\E\paren{m\paren{\vec{X}}m\paren{\vec{X'}}}}.
\end{equation}

Note that $\E(X_1)=2p_0-1$, $\E(X_1')=2p'-1$, $\var(X_1)=4p_0-4p_0^2$, $\var(X_1')=4p'-4p'^2$, and $\cov(X_1,X_1')=4p'-4p_0p'$. Next, for $i=1,\hdots,n$, letting $Y_i = (X_i,X_i')^T$, we can write
\begin{equation}
    \label{eq:variance_yi_2}
    \E(Y_i)=\begin{pmatrix} 2p_0-1 \\ 2p'-1 \end{pmatrix}=\mu_2, \; \var\paren{Y_i}=\begin{bmatrix}4p_0-4p_0^2 & 4p'-4p_0p' \\
    4p'-4p_0p' & 4p'-4p'^2
    \end{bmatrix}=\Sigma_2. 
\end{equation}

Also, let $\bar{Y}_n = \sum_{i=1}^n Y_i/n = (\bar{Y}_{n1},\bar{Y}_{n2})$. An application of the multivariate central limit theorem implies that, as $n\to\infty$,
\begin{equation}
    \label{eq:clt_2}
    \sqrt{n}\paren{\bar{Y}_n-\mu_2} \convD \gauss_2\paren{0,\Sigma_2}.
\end{equation}

In other words, for large $n$, $(\bar{Y}_{n1},\bar{Y}_{n2})$ is approximately distributed as $\gauss_2(\mu_2,\Sigma_2/n)$. On the other hand, $\E(m(\vec{X})m(\vec{X'}))$ can be evaluated as $\E[\sgn(\bar{Y}_{n1})\sgn(\bar{Y}_{n2})]=2\P[\sgn(\bar{Y}_{n1}) = \sgn(\bar{Y}_{n2})]-1$. Combining it with \cref{eq:m=m'}, straightforward calculation leads to the following:
\begin{equation}
    \label{eq:equality_of_majority_2}
    \P\paren{m\paren{\vec{X}}=m\paren{\vec{X'}}} = \P\paren{\sgn(\bar{Y}_{n1}) = \sgn(\bar{Y}_{n2})} = f\paren{\mu_2,\Sigma_2,n},
\end{equation}
where $f(\mu_2,\Sigma_2,n)$ is the probability that the two components of a $\gauss_2(\mu_2,\Sigma_2/n)$ distribution are of same sign. This completes the discussion for the second case. 

For the third part of the theorem, defining $Y_i$ similarly as before, we can obtain
\begin{equation}
    \label{eq:variance_yi_3}
    \E(Y_i)=\begin{pmatrix} 2p_0-1 \\ 2p'-1 \end{pmatrix}=\mu_3, \; \var(Y_i)=\begin{bmatrix}4p_0-4p_0^2 & 4p_0-4p_0p' \\
    4p_0-4p_0p' & 4p'-4p'^2
    \end{bmatrix}=\Sigma_3. 
\end{equation}

Then, the rest of the proof follows in an identical fashion as in the second case.
\end{proof}

\begin{proof}[Proof of \Cref{thm:confidence-interval}]
We have the following distributional convergence for $\hat{p}'$: \begin{equation}
    \label{eq:convergence_p'}
      \frac{\sqrt{n}\paren{\hat{p}'-p'}}{\sqrt{p'\paren{1-p'}}} \convD \gauss\paren{0,1} .
\end{equation}

We know that $\hat{p}' \convP p' \text{ as } n\to\infty$ from the Weak Law of Large Numbers (WLLN). Using Slutsky's theorem, we get the following convergence equation:
\begin{equation}
    \label{eq:final_convergence_p'}
    \frac{\sqrt{n}\paren{\hat{p}'-p'}}{\sqrt{\hat{p}'\paren{1-\hat{p}'}}} \convD \gauss\paren{0,1},
\end{equation}
which implies
\begin{equation}
    \label{eq:CI_p'}
    \P\paren{p' \in S_1\paren{\hat{p}'}}=1-\beta.
\end{equation}

From \Cref{asm-originalprop}, we have the following,
\begin{equation}
\label{eq:assumption analysis}
\nonumber
    \begin{split}
        \int_{0}^{1-p'}\paren{\frac{\pi_0-\pi}{1-\pi}}h\paren{\pi}d\pi=0 & \implies \int_{0}^{1-p'}\paren{1-\frac{1-\pi_0}{1-\pi}}h\paren{\pi}d\pi=0 \\
        & \implies p_0\int_{0}^{1-p'}h\paren{\pi}d\pi=\int_{0}^{1-p'}\paren{\frac{p_0(1-\pi_0)}{1-\pi}}h\paren{\pi}d\pi \\
        & \implies p_0=\frac{1}{H\paren{1-{p}'}}\int_{0}^{1-{p}'}\frac{{p}'}{1-\pi}h\paren{\pi}d\pi.
    \end{split}
\end{equation}

In the above deduction, we have used the fact that $p_0(1-\pi_0)=p'$. Thus, \Cref{asm-originalprop} implies that $p_0=\phi(p')$ i.e.\ $p_0$ matches with the expected value of $p$ when the final proportion of voters voting for the first candidate is ${p}'$. It can be easily shown that $\phi(.)$ is a continuously differentiable function. Applying Delta Theorem on \eqref{eq:convergence_p'}, we get the following,
\begin{equation}
    \label{eq:convergence_P(p')}
       \frac{\sqrt{n}\paren{\phi\paren{\hat{p}'}-\phi\paren{p'}}}{\sqrt{p'\paren{1-p'}\abs{\phi'\paren{p'}}}} \convD \gauss\paren{0,1} .
\end{equation}

We also have the following convergence in probability, 
\begin{equation}
    \label{eq:convergence_P(p')_prob}
    \frac{\sqrt{p'\paren{1-p'}}\abs{\phi'\paren{{p'}}}}{\sqrt{\hat{p}'\paren{1-\hat{p}'}}\abs{\phi'\paren{\hat{p}'}}} \convP 1  .
\end{equation}

Multiplying \eqref{eq:convergence_P(p')} and \eqref{eq:convergence_P(p')_prob} by Slutsky's theorem, we get,
\begin{equation}
    \label{eq:final_convergence_P(p')}
    \frac{\sqrt{n}\paren{\phi\paren{\hat{p}'}-\phi\paren{p'}}}{\sqrt{\hat{p}'\paren{1-\hat{p}'}\abs{\phi'\paren{\hat{p}'}}}} \convD \gauss\paren{0,1}. 
\end{equation}

Hence, from \eqref{eq:final_convergence_P(p')}, 
\begin{equation}
    \label{eq:CI_P(p')}
    \P\paren{\phi\paren{p'} \in S_2\paren{\hat{p}'}}=1-\beta .
\end{equation}

Next, because of the choice of $\beta$, we get the following equation:
\begin{equation}
    \label{eq:CI_P(p')_CI_p'}
    \P\paren{p' \in S_1\paren{\hat{p}'}, \phi\paren{p'} \in S_2\paren{\hat{p}'}} \geqslant \paren{1-\beta}^2=1-\alpha .
\end{equation}

Let $m,M$ be defined as in \eqref{eq:defn_m1_M1}. Using \Cref{asm-originalprop} and the above, 
\begin{equation}
    \label{eq:main CI}
    \nonumber
    \begin{split}
        \P\paren{m \leqslant \eta\paren{p_0,p',n} \leqslant M} & =  \P\paren{m \leqslant \eta\paren{\phi\paren{p'},p',n} \leqslant M} \\
    & \geqslant \P\paren{p' \in S_1\paren{\hat{p}'}, \phi\paren{p'} \in S_2\paren{\hat{p}'}} \\
    & \geqslant 1-\alpha .
    \end{split}
\end{equation}

Thus, $(m,M)$ is a $100(1-\alpha) \%$ confidence interval for $\eta(p_0,p',n)$ and hence for $f(\mu_2,\Sigma_2,n)$. This completes the proof of the confidence interval part of the theorem. 

To prove the consistency of the test, consider the following definitions of $\Theta_0^n$ and $\Theta_1^n$:  
\begin{equation}
    \Theta_0^n=\left\{p_0,p'|\eta(p_0,p',n) \geqslant \tau_c \right\},   \qquad      \Theta_1^n=\left\{p_0,p'|\eta(p_0,p',n) < \tau_c \right\}.
\end{equation}

The region $\Theta_0^n$ corresponds to the null hypothesis $H_0^n$ and the region $\Theta_1^n$ corresponds to the alternate hypothesis $H_1^n$. Electoral intervention is termed as ``significant'' if there is a high probability $\paren{>\tau_c}$ of the majority being changed on performing the intervention. 

We know that if any value lies outside the confidence interval mentioned in \Cref{thm:confidence-interval}, that value is rejected at the level of significance $\alpha$. Keeping this in mind, we define our test statistic to be $M$, as defined in the confidence interval part of \Cref{thm:confidence-interval}. We shall reject $H_0^n$ if $M < \tau_c$. Let us calculate the type-$1$ error of this test. Suppose, $\paren{p_0,p'} \in \Theta_0^n$ i.e.\ $\eta(p_0,p',n) \geqslant \tau_c$. Hence, we have, 
\begin{equation}
\P_{p_0,p'}\paren{H_0^n\text{ is rejected}}=\P_{p_0,p'}\paren{M < \tau_c}=\P_{p_0,p'}\paren{\eta(p_0,p',n) \notin \paren{m,M}}\leqslant \alpha.  
\end{equation} 

Since this is true for all $(p_0,p') \in \Theta_0^n$, we can say the following:
\begin{equation}
  \label{eq:type-1}
     \sup_{\paren{p_0,p'} \in \Theta_0^n}\P_{p_0,p'}\paren{H_0^n\text{ is rejected}} \leqslant \alpha
\end{equation}
 
Thus, we have shown that the maximum type-$1$ error is bounded by $\alpha$. 

Note that the theoretical probability is given by $\P(m(\vec{X})=m(\vec{X'}))$ = $f(\mu_2,\Sigma_2,n)$, and as $n \to\infty$, $f(\mu_2,\Sigma_2,n) \to f(\mu_2,\Sigma_2,\infty)$. Now, $f(\mu_2,\Sigma_2,\infty)=\P(\sgn(P)=\sgn(Q))$ where $(P,Q) \sim \gauss_2(\mu_2,0)$. In other words, $P=\mu_{2,1}$ and $Q=\mu_{2,2}$ almost surely. Since $\mu_2=(2p_0-1,2p'-1)$, $P=2p_0-1$ and $Q=2p'-1$ almost surely. Thus, we conclude that, $\eta(p_0,p',n) \to 1$ as $n \to \infty$ if $(2p_0-1)(2p'-1)>0$ and $\eta(p_0,p',n) \to 0$ as $n \to \infty$ if $(2p_0-1)(2p'-1)<0$. It implies that 
\begin{equation}
    \Theta_0^\infty=\left\{p_0,p'|(2p_0-1)(2p'-1)>0\right\}, \; \Theta_1^\infty=\left\{p_0,p'|(2p_0-1)(2p'-1)<0\right\}.
\end{equation} 

We have already discussed that as $n\to\infty$, $\hat{p}' \convP p'$ and $\phi(\hat{p}') \convP  p$. Hence, as $n\to\infty$,
\begin{equation}
    \label{eq:set convergence}
    S_1\paren{\hat{p}'} \to \left\{p'\right\}, \;
    S_2\paren{\hat{p}'} \to \left\{p_0\right\}.
\end{equation}

Thus, in this case, $m=M=\eta(p_0,p',\infty)$. Under $\Theta_0^\infty$, $(2p_0-1)(2p'-1)>0$ and hence, $M=1$. Under $\Theta_1^\infty$, $(2p_0-1)(2p'-1)<0$ and hence, $M=0 \leqslant \theta$. Clearly, under $\Theta_1^\infty$, the test rejects $\Theta_0^\infty$ with probability 1, and that proves the consistency of the test. 
\end{proof}

\begin{proof}[Proof of \Cref{thm:exit poll theorem}]
For large $n$, $k\hat{p}_k \sim \bin(k,p_0)$. A straightforward application of central limit theorem suggests
\begin{equation}
    \label{eq:convergence_pj}
      \frac{\sqrt{n}\paren{\hat{p}_k-p_0}}{\sqrt{p_0\paren{1-p_0}}} \convD \gauss\paren{0,1} .
\end{equation}

From the WLLN, we know that $\hat{p}_k \convP p_0$ as $k\to\infty$. Using Slutsky's theorem, 
\begin{equation}
    \label{eq:final_convergence_pj}
    \frac{\sqrt{n}\paren{\hat{p}_k-p_0}}{\sqrt{\hat{p}_k\paren{1-\hat{p}_k}}} \convD \gauss\paren{0,1} .
\end{equation}

Subsequently, we have the following:
\begin{equation}
    \label{eq:CI_pj}
    \P\paren{p_0 \in C_2({\hat{p}_k})}=1-\beta .
\end{equation}

Thereafter, one can mimic the steps as in \Cref{sec:final data available} to show that $\P_{H_0}(\mathscr{T}_1<\tau_c) \leqslant \alpha$. 

Next, for showing that the test is consistent, note that $k \to n,\; n \to \infty$ implies that $k \to \infty$ and $\hat{p}_k\convP p_0$, and in that case, from the WLLN, $C_2(\hat{p}_k) \to  \left\{p_0\right\}$. Then, one can adopt identical steps as in the proof of the consistency part in \Cref{thm:confidence-interval} and obtain the required result.
\end{proof}




\end{document}